\documentclass[11pt]{article}

\usepackage[nospace, compress]{cite}
\usepackage{latexsym}

\usepackage{amsmath}
\usepackage{amsthm}
\usepackage{times}
\usepackage{rotating}
\usepackage{multirow}
\usepackage{amssymb,amsfonts,textcomp}
\usepackage{verbatim}
\usepackage{svn}
\usepackage[noend,ruled]{algorithm2e}
\usepackage{url}

\usepackage{cite}

\newcommand{\fixme}[1]{}

\title{Achieving Fully Proportional Representation is Easy in Practice}

\newtheorem{theorem}{Theorem}
\newtheorem{definition}{Definition}

\author{
Piotr Skowron\\
       {University of Warsaw}\\
       {Warsaw, Poland}\\
\and
Piotr Faliszewski\\
       {AGH University}\\
       {Krakow, Poland}\\
\and
Arkadii Slinko\\
       {University of Auckland}\\
       {Auckland, New Zealand}\\
}

\newcommand{\CC}{{{\mathrm{CC}}}}
\newcommand{\Monroe}{{{\mathrm{Monroe}}}}

\newcommand{\score}{{{\mathrm{score}}}}
\newcommand{\bests}{{{\mathrm{bests}}}}
\newcommand{\pos}{{{{\mathrm{pos}}}}}
\newcommand{\borda}{{{{\mathrm{Borda}}}}}
\newcommand{\w}{{{{\mathrm{W}}}}}

\newcommand{\opt}{{{\mathrm{opt}}}}
\newcommand{\ideal}{{{\mathrm{ideal}}}}

\newcommand{\np}{{\mathrm{NP}}}

\newcommand{\naturals}{{{\mathbb{N}}}}

\begin{document}
\maketitle

\begin{abstract}
  We provide experimental evaluation of a number of known and new
  algorithms for approximate computation of Monroe's and
  Chamberlin-Courant's rules. Our experiments, conducted both on
  real-life preference-aggregation data and on synthetic data, show
  that even very simple and fast algorithms can in many cases find
  near-perfect solutions. Our results confirm and complement very
  recent theoretical analysis of Skowron et al., who have shown good
  lower bounds on the quality of (some of) the algorithms that we
  study.
\end{abstract}

\section{Introduction}\label{sec::introduction}

Many countries are governed using indirect democracy, where the people
do not make decisions directly, but rather select representatives
(e.g., a parliament, a senate, a congress) who rule in their
interest. Unfortunately, relatively little effort was so far invested
in the algorithmic study of procedures for electing committees of
representatives (few exceptions include
papers~\cite{mei-pro-ros-zoh:j:multiwinner,pro-ros-zoh:j:proportional-representation,budgetSocialChoice,fullyProportionalRepr}).
Here, we consider two particularly appealing rules for electing
a set of representatives, namely those of Monroe and of Chamberlin and Courant, and
we argue that while these rules in the worst case scenario may be difficult to compute~\cite{pro-ros-zoh:j:proportional-representation,budgetSocialChoice},
in practice, very simple and efficient algorithms find almost-perfect
approximate results.

There are several ways in which countries can choose their parliaments
(or, more generally, in which societies can choose committees of
representatives).  Often, voters are divided into districts and in
each district we hold a local election. For the case of
single-representative districts, in each district we have a
single-winner election held according to one of the standard,
well-known, rules such as the Plurality rule or Borda's rule. In
particular, if the Plurality rule is used then this system is known as
First-Past-the-Post (FPP): In each district the candidate supported by
the largest number of voters is elected.  However, FPP has a number
of drawbacks. For example, it is possible that in a country with two
major parties, $A$ and $B$, even if $49\%$ of the citizens support
party $B$, only members of party $A$ enter the parliament (this
happens if in each district party $A$ has a slight advantage over
party $B$).  Indeed, under FPP the election organizers are
particularly tempted to tamper with the partition of voters into
districts. To circumvent this problem one might use
multi-representative districts, where elections are held using some
multi-winner voting rule (e.g., using Single Transferable Vote
(STV),\footnote{STV for more than one winner is sometimes referred to
  as ``Alternative Vote'' (AV). In a recent referendum Great Britain
  rejected AV as a method for choosing its members of parliament.} or
using a voting rule that assigns scores to the candidates and picks a
group of those with highest scores). However, this approach only
partially solves the problem. Further, compared to
single-representative districts, multi-representative districts loosen
the connection between the candidates and the voters that have elected them.

Fortunately, there is a very attractive way to avoid the problems
mentioned above: Instead of using fixed districts, we may partition
the voters dynamically, based on the votes that they cast. Indeed,
this is exactly the idea behind the fully proportional representation
rules of Monroe~\cite{mon:j:monroe} and of Chamberlin and
Courant~\cite{cha-cou:j:cc}. If we seek a parliament of $K$
representatives, then Monroe's rule says that we should pick a set of
$K$ candidates for whom there is an assignment of these candidates to
the voters such that:
(a) each candidate is assigned to roughly the same number of voters,
and
(b) the total satisfaction of the voters (measured in one of the ways
introduced later) is maximal. Chamberlin-Courant's rule is similar
except that it allows each selected candidate to be matched to a
different number of voters. (Thus if one were to elect a parliament
using Chamberlin-Courant's rule then one should use weighted voting
within the parliament, weighted by the number of voters matched to
each representative.)

In the above description we focus on political elections, but we
mention that both Monroe's rule and Chamberlin-Courant's rule have
many different applications as well. For example, Skowron et
al.~\cite{sko-fal-sli:w:multiwinner} have presented several
(multi-agent) resource allocation settings that can be modeled using
these rules and Lu and Boutilier~\cite{budgetSocialChoice} have
proposed to use Chamberlin-Courant's rule for constructing
recommendations for groups of agents.

Unfortunately, computing Monroe's and Chamberlin-Courant's rules is
both
$\np$-hard~\cite{pro-ros-zoh:j:proportional-representation,budgetSocialChoice}
and difficult in the parametrized sense~\cite{fullyProportionalRepr}.
Thus using them in practice might simply be impossible. The goal of
this paper is to show that not all is lost. We provide experimental
evaluation of a number of known and new algorithms for approximate
computation of Monroe's and Chamberlin-Courant's rules. Our
experiments, conducted both on real-life preference-aggregation data
and on synthetic data, show that even very simple and fast algorithms
can in many cases find near-perfect solutions. Our results confirm and
complement very recent theoretical analysis of Skowron et
al.~\cite{sko-fal-sli:w:multiwinner}, who have shown good lower bounds
on the quality of (some of) our algorithms.
While for single-winner rules using approximate algorithms may be
debatable, for the case of electing a large body of representatives,
e.g., a parliament, using approximation
algorithms seems far easier to justify. Indeed, a good approximate
solution for Monroe's or Chamberlin-Courant's rule represents the
society almost as well as a perfect solution would.

The paper is organized as follows. In Section~\ref{sec:prelim} we 
formally define Monroe's and Chamberlin-Courant's rules. In
Section~\ref{sec:algorithms} we give an overview of the algorithms
that we evaluate and in Section~\ref{sec:DataSets} we describe the
data sets that we use in our experiments. Section~\ref{sec:experiments}
contains our main results. We conclude in Section~\ref{sec:conclusions}.

\section{Preliminaries}\label{sec:prelim}
In this section we briefly review basic notions regarding social
choice theory and we define Monroe's~\cite{mon:j:monroe} and
Chamberlin and Courant's~\cite{cha-cou:j:cc} proportional
representation systems.  We assume the reader is familiar with
standard notions regarding algorithms. For each positive integer $n$,
by $[n]$ we mean the set $\{1, \ldots, n\}$.\medskip

\noindent\textbf{Elections.}
We consider elections over a given set $A = \{a_1, \ldots, a_m\}$ of
alternatives. We have a set $N = [n]$ of agents (the voters), where
each voter $i$, $1 \leq i \leq n$, has a preference order $\succ_i$
over $A$.  A preference order of an agent $i$ is a linear order over
the set $A$; the maximal element is this agent's most preferred
alternative, the minimal element is this agent's least preferred
alternative, and the alternatives in the middle represent the agent's
spectrum of preference. We refer to the collection $V = (\succ_1,
\ldots, \succ_n)$ as the preference profile for a given election.

Let us fix an agent $i$, $1 \leq i \leq n$, and an alternative $a \in
A$. By $\pos_i(a)$ we mean the position $a$ has in $i$'s preference
order. If $a$ is $i$'s most preferred candidate then $\pos_i(a) = 1$,
and if $a$ is $i$'s least preferred candidate then $\pos_i(a) =
\|A\| = m$.\medskip

\noindent\textbf{Positional Scoring Functions.}
Let $m$ be the number of candidates in eleciton. A \emph{positional
  scoring function (PSF)} is any function $\alpha \colon [m]
\rightarrow \naturals$ that satisfies the following two conditions:
(a) $\alpha(m) = 0$, and (b) for each $i, j$, $1 \leq i < j \leq m$,
$\alpha(i) \geq \alpha(j)$.  In Monroe's and in Chamberlin-Courant's
proportional representation rules we will match agents to the
alternatives that represent them.  Intuitively, $\alpha(i)$ is the
amount of satsifaction that an agent derives from being represented by
an alternative that this agent ranks on the $i$'th position.  In this
paper we focus on Borda count PSF, which for $m$ alternatives is
defined as $\alpha^m_{\borda}(i) = m - i$. However, occasionally we
will consider other PSFs as well.

In our algorithms we assume that the PSF $\alpha$ to be used is given
explcitly, as a vector $(\alpha_1, \ldots, \alpha_m)$ of integers such
that for each $i$, $1 \leq i \leq m$, $\alpha(i) = \alpha_i$.
We will implcitly assume that the number of alternatives matches
the domain of the given PSF.
\medskip

\noindent\textbf{Proportional Representation.} Let $A = \{a_1, \ldots, a_m\}$ be
the set of alternatives and $N = [n]$ be the set of agents (with
preference orders over $A$). A representation function is any function
$\Phi \colon N \rightarrow A$. For an $m$-candidate PSF $\alpha$ and a
representation function $\Phi$, $\Phi$'s satisfaction is defined as:
\[
\alpha(\Phi) = \sum_{i=1}^{n} \alpha( \pos_i(\Phi(i))).
\]
Let $K$ be a positive integer. A $K$-$\CC$-representation function is any
representation function $\Phi$ such that $\|\Phi^{-1}(N)\| \leq K$ (that
is, any representation function that matches voters to at most $K$
alternatives). A $K$-$\Monroe$-representation function $\Phi$ is any
$K$-$\CC$-representation function that additionally satisfies the
following requirement: For each $a \in A$ it holds that either
$\lfloor \frac{n}{K} \rfloor \leq \|\Phi^{-1}(a)\| \leq \lceil
\frac{n}{K} \rceil$ or $\|\Phi^{-1}(a)\| = 0$ (that is, each
alternative represents either roughly $\frac{n}{K}$ agents or none of
them).

We will also consider partial representation functions. A partial
$\CC$-representation function is defined in the same way as a regular
one, except that it may assign a null alternative, $\bot$, to some of
the agents.  By convention, we take that for each agent $i$ we have
$\pos_i(\bot) = m$.  A partial $\Monroe$-representation function is
defined analogously: It may assign the null alternative to some voters
(there are no constraints on the number of agents to whom the null
alternative is assigned) but it must be possible to extend it to a
regular $\Monroe$-representation function by replacing the occurrences
of the null alternative with the real ones.

We now define Monroe's and Chamberlin-Courant's (CC) rules.
\begin{definition}
  Let $R$ be a member of $\{\Monroe, \CC\}$. Let $A = \{a_1, \ldots,
  a_m\}$ be a set of alternatives, $N = [n]$ be a set of agents, and
  $\alpha$ be an $m$-candidate PSF. Let $K$ be the size of the set of
  representatives that we seek ($K \leq m$). We say that a $K$-element
  set $W$, $W \subseteq A$, is a set of $\alpha$-$R$ winners if there
  exists a $K$-$R$-representation function $\Phi: N \rightarrow W$
  such that for every other $K$-$R$-representation function $\Psi$ it
  holds that $\alpha(\Phi) \geq \alpha(\Psi)$.
\end{definition}

We point out that for both Monroe's and Chamberlin-Courant's rule
there may be several different winner sets and that some form of
tie-breaking should be applied in these settings. Here we disregard
tie-breaking and simply are interested in \emph{some} winner set (and,
not being able to compute that, in any set with as high a satsifaction as
possible).

It is well-known that for many natural families of PSFs, both for
Monroe's rule and for Chamberlin-Courant's rule, it is $\np$-complete
to decide if there exists a winner set that achieves a given
satisfaction~\cite{pro-ros-zoh:j:proportional-representation,budgetSocialChoice,fullyProportionalRepr}.
However, for each $R$ in $\{\Monroe, \CC\}$, for each PSF $\alpha$
(with the domain matching the number of alternatives in the election),
and for each set $S$ of up to $K$ alternatives we can compute in
polynomial time a (possibly partial) $K$-$R$-representation function
$\Phi^S_R$ that maximizes the agent satisfaction under the condition
that agents are matched to the alternatives in $S$ only. Indeed, it is
easy to see that for $\alpha$-$\CC$ this function is:
\begin{align*}
\Phi^S_{\CC}(i) = \mathrm{argmin}_{a \in S} \pos_i(a)
\end{align*}
and that it is never a partial representation function.  For the case
of $\alpha$-Monroe, computing $\Phi^S_\Monroe$ is more involved and
requires solving a certain min-cost/max-flow problem (see the work of
Betzler et al.~\cite{fullyProportionalRepr}; here if $\|S\| < K$ then
$\Phi_\Monroe^S$ is a partial $\Monroe$-representation function). One
can see that for a given set $S$, there may be many different
(partial) $K$-$\Monroe$-representation functions that achieve optimal satisfaction;
when we write $\Phi_\Monroe^S$, we mean, w.l.o.g., the particular one
computed by the algorithm of Betzler et
al.~\cite{fullyProportionalRepr}.

\section{Algorithms}\label{sec:algorithms}

Let us now describe the algorithms that we will consider in this work.
Some of our algorihtms can be applied both to Monroe's rule and to
Chamberlin-Courtant's rule, while some are specific to only one of
them. For each algorithm we will exactly specify for which rules it is
applicable and, if it is applicable to both, what are the differences.

While most of the algorithms described below are based on ones already
given in the literature, in a number of cases we added heuristics on
top of existing algorithms (which proved to be quite effective, as we
will see later) and, in one case, provided a completely new theortical
analysis. For each algorithm we will carefully describe what was
already known in the literature, and which additions are due to this
paper.

Throughout this section we assume we are given the following setting.
$A = \{a_1, \ldots, a_m\}$ is a set of alternatives, $\alpha$ is an
$m$-candidate PSF, $N = [n]$ is a set of agents, each with a
preference order over $A$, and $K$ is a positive integer, $K \leq m$
(the size of the committee we want to elect).

\subsection{ILP Formulation (Monroe and CC)}
To measure the quality of our approximation algorithms, we compare
their results against optimal solutions that we obtain using integer
linear programs (ILPs) that describe Monroe's and Chamberlin-Courant's
rules. An ILP for Chamberlin-Courant's rule, for arbitraty PSF $\alpha$,
was provided by Lu and Boutilier~\cite{budgetSocialChoice};
the analogous formulation for Monroe's rule was provided by Potthoff and Brams~\cite{potthoff-brams}.
We used the GLPK 4.47 package (GNU Linear Programming Kit, version 4.47)
to solve these ILPs, whenever it was possible to do so in reasonable
time.

\subsection{Algorithms A, B, and C (Monroe)}
Skowron et al.~\cite{sko-fal-sli:w:multiwinner} have suggested and
studied the following algorithm for Monroe's rule, which we will call
Algorithm~A. We start with an empty partial $\Monroe$-representation
function $\Phi$ and we execute $K$ iterations. In each iteration we
do the following:
\newcounter{Lcount}
\begin{list}{\arabic{Lcount}.}{\setlength{\leftmargin}{8pt} \setlength{\labelwidth}{0pt}}
        \setlength{\itemsep}{2pt}
\addtocounter{Lcount}{1}
\item For each alternative $a \in A$ that does not yet represent any
  agents, we compute the maximal satisfaction that some
  not-yet-represented $\lceil \frac{n}{K} \rceil$ agents derive from
  being represented by $a$ (we call this number $\score(a)$ and we refer
  to these agents as $\bests(a)$).
\addtocounter{Lcount}{1}
\item We pick an alternative $a$ with maximum $\score(a)$ and extend
  $\Phi$ by assigning $a$ to represent agents in $\bests(a)$.
\end{list}
This algorithm clearly works in polynomial time. Skowron et
al.~\cite{sko-fal-sli:w:multiwinner} have shown that for
$\alpha^m_{\borda}$ it finds a solution whose satisfaction is at least
a $(1 - \frac{K-1}{2(m-1)} - \frac{H_K}{K})$ fraction of a (possibly
nonexistent) perfect solution, where each agent is represented by his
or her top preference ($H_K$ is the $K$'th harmonic number, i.e., $H_K
= \sum_{i=1}\frac{1}{i} = \Theta(\log K)$).  This suggests that the
algorithm performs best in elections where the size of the committee
we seek is relatively small with respect to the number of
alternatives.

Based on Algorithm~A we have derived Algorithm~B. The only difference
is that after completing the operation of Algorithm~A, we take the set
$S$ of alternatives that were assigned to represent some agents by
Algorithm~A, and replace function $\Phi$ with function
$\Phi_\Monroe^S$, that optimally reassigns the alternatives to the
voters. This very simple heuristic turned out to noticeably improve
the results of the algorithm in practice (and, of course, the
approximation guarantees carry over from Algorithm~A to
Algorithm~B).

Algorithm~C is a further heuristic improvement over Algorithm~B.  This
time the idea is that instead of keeping only one partial function
$\Phi$, we keep a list of up to $d$ partial representation functions,
where $d$ is a parameter of the algorithm. At each iteration, given
these $d$ partial representation functions, for each $\Phi$ of them
and for each alternative $a$ that does not yet have agents assigned to
by this $\Phi$, we compute an optimal extension of this $\Phi$ that
assigns agents to $a$. As a result we obtain possibly more than $d$
(partial) representation functions. For the next iteration we keep
those $d$ of them that have highest satisfaction.

We provide pseudocode for Algorithm~C in
Figure~\ref{alg:greedyImpr}. If we take $d=1$, we obtain
Algorithm~B. If we also disregard the last two lines prior to
returning solution, we obtain Algorithm~A.

\SetKwInput{KwNotation}{Notation}
\NoCaptionOfAlgo
\begin{algorithm}[t]
   \small
   \SetAlCapFnt{\small}
\KwNotation{$\Phi \leftarrow$ a map defining a (partial) representation function, iteratively built by the algorithm. \\
          $\hspace{37pt}$ $\Phi^{\leftarrow} \leftarrow$ the set of agents already represented by some 
           alternative \\
          $\hspace{37pt}$ $\Phi^{\rightarrow} \leftarrow$ the set of alternatives already used in the          
          representation function.\\
          $\hspace{37pt}$ $Par$ $\leftarrow$ a list of partial representation functions}
\vspace{5pt}
   $Par = []$ \\
   $Par$.push($\{\}$) \\
   
   \For{$i\leftarrow 1$ \KwTo $K$}{
      $newPar = []$ \\
      \For{$\Phi \in Par$}{
          $score \leftarrow \{\}$ \\
          $bests \leftarrow \{\}$ \\
          \ForEach{$a_{i} \in A \setminus \Phi^{\rightarrow}$}{
              $agents \leftarrow$ sort $N \setminus \Phi^{\leftarrow}$ so that $j \prec k$ in $agents$ \\
                               $\hspace{37pt}$ $\implies$ $pos_{j}(a_{i}) \leq pos_{k}(a_{i})$ \\
              $bests[a_{i}] \leftarrow$ chose first $\lceil \frac{N}{K} \rceil$ elements of $agents$ \\
              $\Phi' \leftarrow \Phi$ \\
              \ForEach{$j \in bests[a_{i}]$}{
                 $\Phi'[j] \leftarrow a_{i}$ \\
              }
              $newPar$.push($\Phi'$) \\
          }
          sort $newPar$ according to descending order of the total satisfaction of the assigned agents  \\
          $Par \leftarrow$ chose first $d$ elements of $newPar$ \\
      }
   }
   \For{$\Phi \in Par$}{
      $\Phi \leftarrow$ compute the optimal representative function using an algorithm of Betzler et al.~\cite{fullyProportionalRepr} for the set of winners $\Phi^{\rightarrow}$
   }
   \Return{the best representative function from $Par$}
   \label{alg:greedyImpr}
	\caption{\small \textbf{Figure~\ref{alg:greedyImpr}:} The pseudocode for Algorithm~C.}
\end{algorithm}

\subsection{Algorithm GM (Monroe and CC)}
Algorithm~GM (greedy marginal improvement) was introduced by Lu and
Boutilier for the case of Chamberlin-Courant's rule.  Here we
generalize it to apply to Monroe's rule as well, and we show that it
is a $1-\frac{1}{e}$ approximation algorithm for $\alpha$-Monroe.  We
point out that this is the first approximation result for Monroe rule
that applies to all PSFs $\alpha$ (approximability results of Lu and
Boutilier~\cite{budgetSocialChoice} did not apply to $\alpha$-Monroe,
and results of Skowron et al.~\cite{sko-fal-sli:w:multiwinner} applied
to Monroe with Borda count PSF only).  For the case of Monroe, the
algorithm can also be viewed as an extension of Algorithm~B.

Let $R$ be one of $\Monroe$ and $\CC$.  The algorithm proceeds as
follows. We start with an emtpy set $S$.  Then we execute $K$
iterations. In each iteration we find an alternative $a$ that is not
assigned to agents yet, and maximizes the value $\Phi_R^{S \cup
  \{a\}}$. (A certain disadvantage of this algorithm for the case of
Monroe is that it requires a large number of computations of
$\Phi_\Monroe^S$, which is a slow process based on min-cost/max-flow
algorithm.) We provide the pseudocode for Algorithm~GM in
Figure~\ref{alg:greedyOptImpr}.

\begin{algorithm}[t]
  \small \SetAlCapFnt{\small} \KwNotation{ $R$ is either $\Monroe$ or
    $\CC$.}
  \vspace{5pt}
  $S \leftarrow \emptyset$ \\
  \For{$i\leftarrow 1$ \KwTo $K$}{
    $a \leftarrow \mathrm{argmax}_{a \in A \setminus S} \alpha(\Phi_R^{S \cup \{a\}})$ \\
    $S \leftarrow S \cup \{a\}$ \\
  } \Return{$\Phi_{M}^{S}$}
   \label{alg:greedyOptImpr}
   \caption{\small \textbf{Figure~\ref{alg:greedyOptImpr}:} Pseudocode for Algorithm~GM.}
\end{algorithm}

\begin{theorem}
  Algorithm~GM is an $(1 - 1/e)$-approximation algorithm for the
  Monroe'e election problem for arbitrary positional scoring functions.
\end{theorem}
\begin{proof}
  The proof is based on the powerful result of Nemhauser et
  al.~\cite{submodular}, who have shown that greedy algorithms achieve
  $1-\frac{1}{e}$ approximation ratio when used to optimize submodular
  functions. Let $A$ be a set of alternatives, $N = [n]$ a set of
  agents with preferences over $A$, $\alpha$ an $\|A\|$-candidate PSF,
  and $K \leq \|A\|$ the number of representatives that we want to
  elect.

  We consider function $z: 2^{A} \rightarrow \naturals$ defined, for
  each set $S$, $S \subseteq A$ and $\|S\| \leq K$, as $z(S) =
  \alpha(\Phi_{\Monroe}^{S})$.  Clearly, $z(S)$ is monotonic (that is, for
  each two sets $A$ and $B$, if $A \subseteq B$ and $\|B\| \leq K$
  then $z(A) \leq z(B)$. The main part of the proof below is to show
  that $z$ is submodular (we provide the definition below).

  Since $\mathrm{argmax}_{S \subset A, \|S\| = K} z(S)$ is the set of
  winners of our election (under $\alpha$-Monroe) and since
  Algorithm~GM builds the solution iteratively by greedily extending
  initially empty set $S$ so that each iteration increases the value
  of $z(S)$ maximally, by the results of Nemhauser et
  al.~\cite{submodular} we get that Algorithm~GM is a
  $(1-\frac{1}{e})$-approximation algorithm.

  Let us now prove that $z$ is submodular.  That is, our goal is to
  show that for each two sets $S$ and $T$, $S \subset T$, and each
  alternative $a \notin T$ it holds that $z(S \cup \{a\}) - z(S) \geq
  z(T \cup \{a\}) - z(T)$. First, we introduce a notion that
  generalizes the notion of a partial set of winners $S$. Let $s: A
  \rightarrow \naturals$ denote a function that assigns a capacity to
  each alternative (i.e., $s$ gives a bound on the number of agents
  that a given alternative can represent). Intuitively, each set $S$,
  $S \subseteq A$, corresponds to the capacity function that assigns
  $\lceil \frac{n}{k} \rceil$ to each alternative $a \in S$ and 0 to
  each $a \notin S$.  Given a capacity function $s$, we define a
  partial solution $\Phi_{\Monroe}^{s}$ to be one that maximizes the
  total satisfaction of the agents and that satisfies the capacity
  constraints: $\forall_{a \in S} \|(\Phi_{\Monroe}^{s})^{-1}(a)\|
  \leq s(a)$. To simplify notation, we write $s \cup \{a\}$ to denote
  the function such that $(s \cup \{a\})(a) = s(a) + 1$ and
  $\forall_{a' \in S} (s \cup \{a\})(a') = s(a')$. (Analogously, we
  interpret $s \setminus \{a\}$ as subtracting one from the capacity
  for $a$; provided it is nonzero.)  Also, by $s \leq t$ we mean that
  $\forall_{a \in A} s(a) \leq t(a)$. We extend our function $z$ to
  allow us to consider a subset of the agents only.  For each subset
  $N'$ of the agents and each capacity function $s$, we define
  $z_{N'}(s)$ to be the satisfaction of the agents in $N'$ obtained
  under $\Phi_{\Monroe}^{s}$. We will now prove a stronger variant of
  submodularity for our extended $z$. That is, we will show that for
  each two capacity functions $s$ and $t$ it holds that:
  \begin{equation}\label{eq:submodularity}
    s \leq t \Rightarrow z_{N}(s \cup \{a\}) -  z_{N}(s) \geq z_{N}(t \cup \{a\}) -  z_{N}(t)
  \end{equation}
  Our proof is by induction on $N$.  Clearly,
  Equation~\eqref{eq:submodularity} holds for $N' = \emptyset$. Now,
  assuming that Equation~\eqref{eq:submodularity} holds for every $N'
  \subset N$ we will prove its correctness for $N$.  Let $i$ denote an
  agent such that $\Phi_{\Monroe}^{t \cup \{a\}}(i) = a$ (if there is
  no such agent then clearly the equation holds). Let $a_{s} =
  \Phi_{\Monroe}^{s}(i)$ and $a_{t} = \Phi_{\Monroe}^{t}(i)$. We have:
  \begin{align*}
    z_{N}(t \cup \{a\}) -  z_{N}(t) =   \alpha(\pos_i(a)) +  z_{N \setminus \{i\}}(t)   
     - \alpha(\pos_i(a_{t})) - z_{N \setminus \{i\}}(t \setminus \{a_{t}\}). 
  \end{align*}
  We also have:
  \begin{align*}
    z_{N}(s \cup \{a\}) -  z_{N}(s)  \geq \alpha(\pos_i(a)) + z_{N \setminus \{i\}}(s)  
     - \alpha(\pos_i(a_{s})) - z_{N \setminus \{i\}}(s \setminus \{a_{s}\}).
  \end{align*}
  Since $\Phi_{\Monroe}^{t}$ describes an optimal representation function under
  the capacity restrictions $t$, we have that:
  \begin{align*}
    \alpha(\pos_i(a_t)) + z_{N \setminus \{i\}}(t \setminus a_{t})
    \geq \alpha(\pos_i(a_s)) + z_{N \setminus \{i\}}(t \setminus \{a_{s}\})
  \end{align*}
  Finally, from the inductive hypothesis for $N' = N \setminus \{i\}$ we have:
  \begin{align*}
    z_{N \setminus \{i\}}(s) - z_{N \setminus \{i\}}(s \setminus \{a_{s}\}) \geq
    z_{N \setminus \{i\}}(t) - z_{N \setminus \{i\}}(t \setminus \{a_{s}\})
  \end{align*}
  By combining these inequalities we get:
  \begin{align*}
    z_{N}(s \cup \{a\}) -  z_{N}(s) & \geq \alpha(\pos_i(a)) + z_{N \setminus \{i\}}(s) 
     - (\alpha(\pos_i(a_{s})) + z_{N \setminus \{i\}}(s \setminus \{a_{s}\})) \\
    & \geq \alpha(\pos_i(a)) - \alpha(\pos_i(a_{s})) 
     +  z_{N \setminus \{i\}}(t) - z_{N \setminus \{i\}}(t \setminus \{a_{s}\}) \\
    & \geq \alpha(\pos_i(a)) + z_{N \setminus \{i\}}(t) 
     - \alpha(\pos_i(a_t)) -  z_{N \setminus \{i\}}(t \setminus \{a_{t}\}) \\
    & = z_{N}(t \cup \{a\}) - z_{N}(t)
  \end{align*}
  This completes the proof.
\end{proof}

\subsection{Algorithm C (CC)}

This algorithm, introduced in this paper, proceeds like Algorithm GM
for Chamberlin-Courant's rule, but in each iteration it keeps up to
$d$ (partial) CC-representation functions $\Phi^S_\CC$, for distinct
subsets $S$ of alternatives ($d$ is a parameter of the algorithm). In
each iteration the algorithm extends each function $\Phi^S_\CC$ by
every possible alternative (obtaining $O(dm)$ new representation
functions) and stores up to $d$ of them, that obtain highest
satisfaction.

\subsection{Algorithm P (CC)}
Algorithm~P (position restriction) was introduced and studied by
Skowron et al.~\cite{sko-fal-sli:w:multiwinner}. The algorithm
proceeds as follows. First, we consider a certain number $x$
(specifically, $x = \lceil \frac{m \w(K)}{K}\rceil$, where $\w(x)$ is
Lambert's $\w$ function, defined as the solution of equality $x =
\w(x)e^{\w(x)}$). Then, the algorithm tries to greedily find a cover
of as many agents as possible with $K$ alternatives (an alternative is
said to cover a given agent if this agent ranks this alternative among
top $x$ positions).  Skowron et al.~\cite{sko-fal-sli:w:multiwinner}
have shown that for $\alpha^m_{\borda}$ this algorithm finds a
solution that is at most $1 - \frac{2\w(K)}{K}$ times worse than a
perfect (possibly nonexistent) solution, where every agent is
represented by his or her top-preferred alternative. The pseudocode
for Algorithm~P is presented in Figure~\ref{alg:posConCC}.

\begin{algorithm}[t]
   \small
   \SetAlCapFnt{\small}
   \KwNotation{We use the same notation as in Algorithm~\ref{alg:greedyImpr}\\
   $\hspace{37pt}$ $\mathrm{num\_pos}_x(a) \leftarrow \|\{i \in [n] \setminus \Phi^{\leftarrow} : pos_i(a) \leq x \}\|$ \\ $\hspace{37pt}$ (the number of not-yet assigned agents that rank alternative $a$ in one of their first $x$ positions)}
\vspace{5pt}
   $\Phi = \{\}$ \\
   $x = \lceil \frac{m\w(K)}{K} \rceil$ \\
   \For{$i\leftarrow 1$ \KwTo $K$}{
                        $a_{i} \leftarrow \mathrm{argmax}_{a \in A \setminus \Phi^{\rightarrow}} \mathrm{num\_pos}_x(a)$

      \ForEach{$j \in [n] \setminus \Phi^{\leftarrow}$}{
         \If{$pos_j(a_{i}) < x$}{
             $\Phi[j] \leftarrow a_{i}$ \\
         }
      }
   }
   \ForEach{$j \in A \setminus \Phi^{\leftarrow}$}{
       $a \leftarrow$ such server from $\Phi^{\rightarrow}$ that $\forall_{a' \in \Phi^{\rightarrow}} pos_{j}(a) \leq pos_{j}(a')$ \\
       $\Phi[j] \leftarrow a$ \\
   }
   \label{alg:posConCC}
   \caption{\small \textbf{Figure~\ref{alg:posConCC}:} Pseudocode for Algorithm~P.}
\end{algorithm}
\addtocounter{figure}{\value{algocf}}

\subsection{Algorithm R (Monroe and CC)}

Algorithm~R (random sampling) is based on picking the set of winners
randomly and matching them optimally to the agents.  Skowron et
al.~\cite{sko-fal-sli:w:multiwinner} have shown that if one chooses a
set $S$ of $K$ alternatives uniformly at random, then for
$\alpha^m_\borda$-Monroe, the expected satisfaction of
$\alpha^m_\borda(\Phi^S_\Monroe)$ is $\frac{1}{2}(1 + \frac{K}{m} -
\frac{K^2}{m^2-m} + \frac{K^3}{m^3-m^2}) - \epsilon$, and that one has
to repeat this process $\frac{-512 \log (1 - \lambda)}{K\epsilon^2}$
times, to reach probability $\lambda$ of achieving this satisfaction.
For example, for $\lambda=0.99$ and $\epsilon=0.1$ this algorithm
would require to repeat the sampling process $340000/K$ times (each
time executing a costly matching algorithm).  This makes the algorithm
impractical, especially for small instances (where $K$ is low). 
Thus in our experimental evaluation we will consider the
modification of the algorithm that repeats the sampling process only
100 times.

Oren~\cite{ore:p:cc} has shown an analogous result for the
case of Chamberlin-Courant's rule.

\subsection{Summary of the Algorithms}
We summarize the algorithms that we use in Table~\ref{tab:algs}.  In
particular, the table clearly shows that for the case of Monroe,
Algorithms B and C are not much slower than Algorithm A but offer a
chance of improved peformance. Algorithm GM is intuitively even more
appealing, but achieves this at the cost of high time complexity. For
the case of Chamberlin-Courant's rule, it is unclear which of the
algorithms to expect to be superior. One of the main goals of this
paper is to establish if either of the presented algorithms clearly
dominates the others. Our implementations are available at
\url{http://mimuw.edu.pl/~ps219737/monroe/experiments.tar.gz}.

\begin{table}[t]

\begin{center}
  \setlength{\tabcolsep}{3pt}
  \begin{tabular}{|c|l|l|c|}
  \hline
  Algorithm & Approximation ratio for Borda PSF & Runtime & Reference \\
  \hline
  A    & $1 - \frac{K-1}{2(m-1)} - \frac{H_K}{K}$ & $Kmn$ & Skowron et al.~\cite{sko-fal-sli:w:multiwinner}  \\
  B    & as in Alg. A & $Kmn$$+$$O(\Phi^{S})$ & (this paper)\\
  C    & as in Alg. A & $dKmn$$+$$dO(\Phi^{S})$ & (this paper)\\
  GM   & as in Alg. A for Borda PSF; $1-\frac{1}{e}$ for others   & $KmO(\Phi^{S})$ & (this paper) \\
  R    &  $\frac{1}{2}(1 + \frac{K}{m} - \frac{K^2m - K^{3}}{m^3-m^2})$ & $\frac{|\log (1 -\lambda)|}{K\epsilon^2}O(\Phi^{S})$ & Skowron et al.~\cite{sko-fal-sli:w:multiwinner} \\
  \hline
  P    & $1 - \frac{2\w(K)}{K}$ & $nm\w(K)$ & Skowron et al.~\cite{sko-fal-sli:w:multiwinner} \\
  GM   & $1-\frac{1}{e}$ & $Kmn$ & Lu and Boutilier~\cite{budgetSocialChoice} \\
  C    & as in Alg. GM & $dKm(n$$+$$\log dm)$ & (this paper)  \\
  R    & $(1-\frac{1}{K+1})(1 + \frac{1}{m})$ & $\frac{|\log (1 -\lambda)|}{\epsilon^2}n$ & Oren~\cite{ore:p:cc} \\
  \hline  
  \end{tabular}

  \caption{\label{tab:algs}A summary of the algorithms studied in this
    paper. The top of the table presents algorithms for Monroe's rule
    and the bottom for Chamberlin-Courant's rule. In column
    ``Approx.'' we give currently known approximation ratio for the
    algorithm under Borda PSF, on profiles with {$m$}
    candidates and where the goal is to select a committee of size
    {$K$}. Here, {$O(\Phi^{S}) = O(n^2(K +
      \mathrm{log}n))$} is the complexity of finding a partial
    representation function with the algorithm of Betzler et
    al.~\cite{fullyProportionalRepr}.}
\end{center}
\end{table}

\section{Experimental Data}\label{sec:DataSets}
We have considered both real-life preference-aggregation data and
synthetic data, generated according to a number of election models.

\subsection{Real-Life Data}\label{sec:real-data}

We have used real-life data regarding people's preference on sushi types,
movies, college courses, and competitors' performance in
figure-skating competitions.
One of the major problems regarding real-life preference data is that
either people express preferences over a very limited set of
alternatives, or their preference orders are partial. To address the
latter issue, for each such data set we complemented the partial
orders to be total orders using the technique of
Kamishima~\cite{Kamishima:Nantonac}.  (The idea is to complete each
preference order based on those reported preference orders that appear
to be similar.)

Some of our data sets contain a single profile, whereas the others
contain multiple profiles.  When preparing data for a given number $m$
of candidates and a given number $n$ of voters from a given data set,
we used the following method: We first uniformly at random chose a
profile within the data set, and then we randomly selected $n$ voters
and $m$ candidates. We used preference orders of these $n$ voters
restricted to these $m$ candidates.

\smallskip



\noindent\textbf{Sushi Preferneces.}\quad
We used the set of preferences regarding sushi types collected by
Kamishima\cite{Kamishima:Nantonac}.\footnote{The sushi data set is
  available under the following url:
  \url{http://www.kamishima.net/sushi/}}  Kamishima has collected two
sets of preferences, which we call \textsc{S1} and \textsc{S2}. Data
set S1 contains complete rankings of $10$ alternatives collected from
$5000$ voters.  S2 contains partial rankings of $5000$ voters over a
set of $100$ alternatives (each vote ranks $10$ alternatives). We used
Kamishima~\cite{Kamishima:Nantonac} technique to obtain total
rankings.\smallskip

\noindent\textbf{Movie Preferences.}\quad
Following Mattei et al.~\cite{Mattei:Netflix}, we have used the NetFlix data
set\footnote{http://www.netflixprize.com/} of movie preferences (we
call it \textsc{Mv}).  NetFlix data set contains ratings collected
from about $480$ thousand distinct users regarding $18$ thousand
movies. The users rated movies by giving them a score between $1$
(bad) and $5$ (good). The set contains about $100$ million ratings.
We have generated $50$ profiles using the following method: For each
profile we have randomly selected $300$ movies, picked $10000$ users
that ranked the highest number of the selected movies, and for each
user we have extended his or her ratings to a complete preference
order using the method of
Kamishima~\cite{Kamishima:Nantonac}.\smallskip

\noindent\textbf{Course Preferences.}\quad
Each year the students at the AGH University choose
courses that they would like to attend. The students are offered a
choice of six courses of which they have to attend three.  Thus the
students are asked to give an unordered set of their three
top-preferred courses and a ranking of the remaining ones (in case too
many students select a course, those with the highest GPA are enrolled
and the remaining ones are moved to their less-preferred courses). In
this data set, which we call \textsc{Cr}, we have $120$ voters
(students) and $6$ alternatives (courses). However, due to the nature
of the data, instead of using Borda count PSF as the satisfaction measure,
we have used the vector $(3,3,3,2,1,0)$. We made this data set
publicly available under the url: \url{http://mimuw.edu.pl/~ps219737/monroe/registration.tar.gz}. \smallskip

\noindent\textbf{Figure Skating.}\quad This data set, which we call \textsc{Sk}, contains
preferences of the judges over the performances in a figure-skating
competitions. The data set contains $48$ profiles, each describing a
single competition. Each profile contains preference orders of $9$
judges over about 20 participants. The competitions include European
skating championships, Olympics, World Junior, and World
Championships, all from 1998\footnote{This data set is available under
  the following url: \url{http://rangevoting.org/SkateData1998.txt}.}.
(Note that while in figure skating judges provide numerical scores,
this data set is preprocessed to contain preference orders.)

\subsection{Synthetic Data}

For our tests, we have also used profiles generated using three
well-known distributions of preference orders. \smallskip

\noindent\textbf{Impartial Culture.}\quad Under impartial culture
model of preferences (which we denote \textsc{IC}), for a given set
$A$ of alternatives, each voter's preference order is drawn uniformly
at random from the set of all possible total orders over $A$.  While
not very realistic, profiles generated using impartial culture model
are a standard testbed of election-related algorithms.  \smallskip

\noindent\textbf{Polya-Eggenberger Urn Model.}\quad Following Walsh~\cite{Walsh11},
we have used the Polya-Eggenberger urn model~\cite{bpublicchoice85} (which
we denote \textsc{Ur}). In this model we generate votes as follows. We
have a set $A$ of $m$ alternatives and an urn that initially contains
all $m!$ preference orders over $A$.  To generate a vote, we simply
randomly pick one from the urn (this is our generated vote), and
then---to simulate correlation between voters---we return $a$ copies
of this vote to the urn. When generating an election with $m$ candidates
using the urn model, we have set the parameter $a$ so that
$\frac{a}{m!} = 0.05$ (Walsh~\cite{Walsh11} calls this parameter $b$;
we mention that other authors use much higher values of $b$ but we felt
that too high a value of $b$ leads to a much too strong a correlation
between votes).\smallskip

\noindent\textbf{Generalized Mallow's Model.}\quad We refer to this
data set as \textsc{Ml}. Let $\succ$ and $\succ'$ be two preference
orders over some alternative set $A$. Kendal-Tau distance between
$\succ$ and $\succ'$, denoted $d_{K}(\succ,\succ')$, is defined as the
number of pairs of candidates $x, y \in A$ such that either $x \succ
y \land y \succ' x$ or $y \succ x \land x \succ' y$.

Under Mallow's distribution of preferences~\cite{mallowModel} we are
given two parameters: A \emph{center} preference order $\succ$ and a
number $\phi$ between $0$ and $1$. The model says that the probability
of generating preference order $\succ'$ is proportional to the value
$\phi^{d_{K}(\succ,\succ')}$.  To generate preference orders following
Mallow's distribution, we use the algorithm given by Lu and Boutilier
\cite{mallowImplementation2011}.

In our experiments, we have used a mixture of Mallow's models.  Let
$A$ be a set of alternatives and let $a$ be a positive integer. This
mixture model is parametrized by three vectors, $\Lambda =
(\lambda_{1}, \dots, \lambda_{a})$ (where each $\lambda_i$, $1 \leq i
\leq a$ is between $0$ and $1$, and $\sum_{i=1}^a\lambda_1=1$), $\Phi
= (\phi_{1}, \dots, \phi_{a})$ (where each $\phi_i$, $1 \leq i \leq
a$, is a number between $0$ and $1$), and $\Pi = (\succ_{1}, \ldots,
\succ_{a})$ (where each $\succ_i$, $1 \leq i \leq a$, is a preference
order over $A$). To generate a vote, we pick a random integer $i$, $1
\leq i \leq a$ (each $i$ is chosen with probability $\lambda_i$), and
then generate the vote using Mallow's model with parameters
$(\succ_i,\phi_i)$.

For our experiments we have used $a = 5$, and we have generated
vectors $\Lambda$, $\Phi$, and $\Pi$ uniformly at random.

\section{Experiments}\label{sec:experiments}

In this section we present the results of the evaluation of algorithms
from Section~\ref{sec:algorithms} on the data sets from
Section~\ref{sec:DataSets}. In all cases, except for the college
courses data set, we have used Borda PSF to measure voter
satisfaction.  For the case of the courses data set, we have used
vector $(3,3,3,2,1,0)$.

We have conducted three sets of experiments. First, we have tested all
our algorithms on relatively small elections (up to $10$ candidates,
up to $100$ agents). In this case we were able to compare our
algorithms' solutions with the optimal ones. (To obtain the optimal
solutions we were using the ILP formulations and GLPK's ILP solver.)
Thus we report the quality of our algorithms as the average of
fractions ${C}/{C_{\opt}}$, where $C$ is the satisfaction obtained by
a respective algorithm and $C_{\opt}$ is the satisfaction in the
optimal solution.
For each algorithm and data set, we also report the average fraction
${C}/{C_{\ideal}}$, where $C_{\ideal}$ is the satisfaction that the
voters would have obtained if each of them were matched to his or her
most preferred alternative. In our further experiments, where we
consider larger elections, we were not able to compute optimal
solutions, but fraction ${C}/{C_{\ideal}}$ gives a lower bound for
${C}/{C_{\opt}}$.  We report this value for small elections so that
we see an example of relation between ${C}/{C_{\opt}}$ and
${C}/{C_{\ideal}}$ and so that we can compare the results for small
elections with the results for the larger ones.  Further, for the case
of Borda PSF the ${C}/{C_\ideal}$ fraction has a very natural
interpretation: If its value is $\alpha$ (for a given solution), then,
on the average, in this solution each voter is matched to an
alternative that he or she prefers to $(m-1)\alpha$ alternatives.

In our second set of experiments we have run our algorithms on large
elections (thousands of agents, hundreds of alternatives), coming
either from the NetFlix data set or generated by us using one of our
models. Here we reported the average fraction ${C}/{C_{\ideal}}$ only.
We have analyzed the quality of the solutions as a function of the
number of agents, the number of candidates, and the relative number of
winners (fraction $K/m$). (This last set of results is particularly
interesting because in addition to measuring the quality of our
algorithms, it allows one to asses the size of a committee one should
seek if a given average agent satisfaction is to be obtained).

In the third set of experiments we have measured running times of our
algorithms and of the ILP solver.

\begin{table}[t]
\begin{center}

\begin{tabular}{|c|c|c|c|c|c||c|c|c|c|}
\cline{2-10}
\multicolumn{1}{c|}{} & \multicolumn{5}{|c||}{Monroe} & \multicolumn{4}{|c|}{CC} \\
\cline{2-10}
\multicolumn{1}{c|}{} & A & B & C & GM & R    & C & GM & P & R\\
\cline{1-10}
\textsc{S1} & $0.94$ & $0.99$        & $\approx 1.0$  & $0.99$         & $0.99$ & $1.0$ & $\approx 1.0$ & $0.99$ & $0.99$ \\
\textsc{S2} & $0.95$ & $0.99$        & $1.0$          & $\approx 1.0$  & $0.99$ & $1.0$ & $\approx 1.0$ & $0.98$ & $0.99$ \\
\textsc{Mv} & $0.96$ & $\approx 1.0$ & $1.0$          & $\approx 1.0$  & $0.98$ & $1.0$ & $\approx 1.0$ & $0.96$ & $\approx 1.0$ \\
\textsc{Cr} & $0.98$ & $0.99$        & $1.0$          & $\approx 1.0$  & $0.99$ & $1.0$ & $\approx 1.0$ & $1.0$  & $\approx 1.0$ \\
\textsc{Sk} & $0.99$ & $\approx 1.0$ & $1.0$          & $\approx 1.0$  & $0.94$ & $1.0$ & $\approx 1.0$ & $0.85$ & $0.99$ \\
\cline{1-10}
\textsc{IC} & $0.94$ & $0.99$        & $\approx 1.0$  & $0.99$         & $0.99$ & $1.0$ & $\approx 1.0$ & $0.99$ & $0.99$\\
\textsc{Ml} & $0.94$ & $0.99$        & $1.0$          & $0.99$         & $0.99$ & $1.0$ & $\approx 1.0$ & $0.99$ & $0.99$ \\
\textsc{Ur} & $0.95$ & $0.99$        & $\approx 1.0$  & $0.99$         & $0.99$ & $1.0$ & $0.99$        & $0.97$ & $0.99$ \\
\cline{1-10}
\end{tabular}
\caption{The average quality of the  algorithms compared with the optimal solution ({$C/C_{\opt}$}) for the small instances of data and for {$K=3$}.}
\label{table:qualityAlgs1}
\end{center}
\end{table}

\begin{table}[t]
\begin{center}
\begin{tabular}{|c|c|c|c|c|c||c|c|c|c|}
\cline{2-10}
\multicolumn{1}{c|}{} & \multicolumn{5}{|c||}{Monroe} & \multicolumn{4}{|c|}{CC} \\
\cline{2-10}
\multicolumn{1}{c|}{} & A & B & C & GM & R    & C & GM & P & R\\
\cline{1-10}
\textsc{S1} & $0.95$ & $\approx 1.0$ & $1.0$          & $0.99$         & $0.99$        & $1.0$ & $\approx 1.0$ & $0.97$ & $0.99$ \\
\textsc{S2} & $0.94$ & $0.99$        & $\approx 1.0$  & $0.99$         & $0.99$        & $1.0$ & $\approx 1.0$ & $0.98$ & $\approx 1.0$ \\
\textsc{Mv} & $0.95$ & $0.99$        & $1.0$          & $\approx 1.0$  & $0.98$        & $1.0$ & $\approx 1.0$ & $0.97$ & $\approx 1.0$ \\
\textsc{Cr} & $0.96$ & $\approx 1.0$ & $1.0$          & $\approx 1.0$  & $0.99$        & $1.0$ & $1.0$         & $1.0$  & $1.00$ \\
\textsc{Sk} & $0.99$ & $\approx 1.0$ & $1.0$          & $\approx 1.0$  & $0.88$        & $1.0$ & $1.0$         & $0.91$ & $\approx 1.0$ \\
\cline{1-10}
\textsc{IC} & $0.95$ & $0.99$        & $\approx 1.0$  & $0.99$         & $0.99$        & $1.0$ & $\approx 1.0$ & $0.99$ & $0.99$\\
\textsc{Ml} & $0.95$ & $0.99$        & $\approx 1.0$  & $0.99$         & $0.99$        & $1.0$ & $\approx 1.0$ & $0.98$ & $0.99$ \\
\textsc{Ur} & $0.96$ & $0.99$        & $\approx 1.0$  & $0.99$         & $\approx 1.0$ & $1.0$ & $\approx 1.0$ & $0.96$ & $0.99$ \\
\cline{1-10}

\end{tabular}

\caption{The average quality of the  algorithms compared with the optimal solution ({$C/C_{\opt}$}) for the small instances of data and for {$K=6$}.}
\label{table:qualityAlgs2}
\end{center}
\end{table}

\begin{table}[t]
\begin{center}
\begin{tabular}{|c|c|c|c|c|c||c|c|c|c|}
\cline{2-10}
\multicolumn{1}{c|}{} & \multicolumn{5}{|c||}{Monroe} & \multicolumn{4}{|c|}{CC} \\
\cline{2-10}
\multicolumn{1}{c|}{} & A & B & C & GM & R    & C & GM & P & R\\
\cline{1-10}
\textsc{S1} & $0.85$ & $0.89$ & $0.9$   & $0.89$ & $0.89$ & $0.92$ & $0.89$ & $0.91$ & $0.92$ \\
\textsc{S2} & $0.85$ & $0.89$ & $0.89$  & $0.89$ & $0.89$ & $0.93$ & $0.9$  & $0.91$ & $0.92$ \\
\textsc{Mv} & $0.88$ & $0.92$ & $0.92$  & $0.92$ & $0.91$ & $0.97$ & $0.92$ & $0.93$ & $0.97$ \\
\textsc{Cr} & $0.94$ & $0.97$ & $0.96$  & $0.96$ & $0.96$ & $0.97$ & $0.97$ & $0.97$ & $0.97$ \\
\textsc{Sk} & $0.96$ & $0.96$ & $0.97$  & $0.97$ & $0.91$ & $1.0$  & $0.97$ & $0.82$ & $0.99$ \\
\cline{1-10}
\textsc{IC} & $0.8$  & $0.84$ & $0.85$  & $0.84$ & $0.84$ & $0.85$ & $0.83$ & $0.84$ & $0.85$\\
\textsc{Ml} & $0.83$ & $0.88$ & $0.88$  & $0.9$  & $0.88$ & $0.92$ & $0.90$ & $0.89$ & $0.94$ \\
\textsc{Ur} & $0.8$  & $0.85$ & $0.86$  & $0.87$ & $0.85$ & $0.9$  & $0.87$ & $0.87$ & $0.89$ \\
\cline{1-10}
\end{tabular}

\caption{The average quality of the  algorithms compared with the simple lower bound ({$C/C_{\ideal}$}) for the small instances of data and for {$K=3$}.\smallskip}

\label{table:qualityAlgs3}
\end{center}
\vspace{-0.5cm}
\end{table}

\begin{table}[t]
\begin{center}
\begin{tabular}{|c|c|c|c|c|c||c|c|c|c|}
\cline{2-10}
\multicolumn{1}{c|}{} & \multicolumn{5}{|c||}{Monroe} & \multicolumn{4}{|c|}{CC} \\
\cline{2-10}
\multicolumn{1}{c|}{} & A & B & C & GM & R    & C & GM & P & R\\
\cline{1-10}
S1 & $0.91$ & $0.96$ & $0.96$  & $0.95$ & $0.95$ & $0.98$ & $0.98$ & $0.96$ & $0.98$ \\
S2 & $0.88$ & $0.93$ & $0.93$  & $0.93$ & $0.93$ & $0.98$ & $0.98$ & $0.96$ & $0.98$ \\
Mv & $0.85$ & $0.89$ & $0.89$  & $0.89$ & $0.88$ & $0.99$ & $0.99$ & $0.97$ & $0.99$ \\
Cr & $0.95$ & $0.98$ & $0.99$  & $0.99$ & $0.98$ & $1.0$  & $1.0$  & $1.0$  & $1.0$ \\
Sk & $0.91$ & $0.92$ & $0.92$  & $0.92$ & $0.81$ & $1.0$  & $1.0$  & $0.91$ & $\approx 1.0$ \\
\cline{1-10}
IC & $0.91$ & $0.95$ & $0.95$  & $0.94$ & $0.95$ & $0.96$ & $0.96$ & $0.95$ & $0.95$\\
Ml & $0.89$ & $0.94$ & $0.94$  & $0.94$ & $0.93$ & $0.97$ & $0.98$ & $0.95$ & $0.98$ \\
Ur & $0.91$ & $0.95$ & $0.95$  & $0.94$ & $0.95$ & $0.98$ & $0.98$ & $0.94$ & $0.97$ \\
\cline{1-10}

\end{tabular}

\caption{The average quality of the  algorithms compared with the simple lower bound ({$C/C_{\ideal}$}) for the small instances of data and for $K=6$.}

\label{table:qualityAlgs4}
\end{center}
\end{table}

\subsection{Evaluation on Small Instances}

We now present the results of our experiments on small
elections. 
%
%
%
For each data set, we generated elections with the number of agents
$n=100$ ($n=9$ for data set \textsc{Sk} because there are only $9$
voters there) and with the number of alternatives $m=10$ ($m=6$ for
data set \textsc{Cr} because there are only $6$ alternatives there)
using the method described in Section~\ref{sec:real-data} for the
real-life data sets, and in the natural obvious way for synthetic
data.  For each algorithm and for each data set we ran $500$
experiments on different instances for $K=3$ (for the \textsc{Cr}
data set we used $K=2$) and $500$ experiments for $K=6$ (for
\textsc{Cr} we set $K=4$). For Algorithms $C$ we set the
parameter $d = 15$. The results (average fractions ${C}/{C_\opt}$ and
${C}/{C_\ideal}$) for $K=3$ are given in
Tables~\ref{table:qualityAlgs1}~and~\ref{table:qualityAlgs3};
the results for $K=6$ are given in
Tables~\ref{table:qualityAlgs2}~and~\ref{table:qualityAlgs4}
(they are almost identical as for $K=3$). For each experiment in this
section we also computed the standard deviation; it was always 
on the order of $0.01$.
The results lead to the following conclusions:

\setcounter{Lcount}{0}
\begin{list}{\arabic{Lcount}.}{\setlength{\leftmargin}{8pt} \setlength{\labelwidth}{0pt}}
        \setlength{\itemsep}{2pt}
\addtocounter{Lcount}{1}
\item Even Algorithm~A obtains very good results, but nonetheless
  Algorithms~B and C improve noticeably upon Algorithm~A. In
  particular, Algorithm~C (for $d=15$) obtains the highest
  satisfaction on all data sets and in almost all cases was able to
  find an optimal solution.
\addtocounter{Lcount}{1}
\item Algorithm~R gives slightly worse solutions than 
  Algorithm~C. 
\addtocounter{Lcount}{1}
\item The quality of the algorithms does not depend on the data set
  used for verification (the only exception is Algorithm~R for
  Monroe's system on data set \textsc{Sk}; however \textsc{Sk} has
  only 9 voters so it can be viewed as a border case).
\end{list}

\subsection{Evaluation on Larger Instances}

For experiments on larger instances we needed data sets with at
least $n = 10 000$ agents. Thus we used the NetFlix data set and
synthetic data. (Additionally, we run the subset of experiments (for
$n \leq 5000$) also for the \textsc{S2} data set.) For Monroe's rule we
present results for Algorithm~A, Algorithm~C, and Algorithm~R, and for
Chamberlin-Courant's rule we present results for Algorithm~C and
Algorithm~R. We limit the set of algorithms for the sake of the clarity of
the presentation. For Monroe we chose Algorithm~A because it is the
simplest and the fastest one, Algorithm~C because it is the best
generalization of Algorithm~A that we were able to run in reasonable
time, and Algorithm~R to compare a randomized algorithm to
deterministic ones. For Chamberlin-Courant's rule we chose Algorithm~C
because it is, intuitively, the best one, and we chose Algorithm~R for
the same reason as in the case of Monroe.  Further, we present results
for the NetFlix data set and for the urn model only. We chose these data sets because the
urn model results turned out to be the worst ones among the synthetic
data sets, and the NetFlix data set is our only large real-life data
set.

First, for each data set and for each algorithm we fixed the value of
$m$ and $K$ and for each $n$ ranging from $1000$ to $10000$ with the
step of $1000$ we run $50$ experiments. We repeated this procedure for
4 different combinations of $m$ and $K$: ($m = 10$, $K = 3$), ($m =
10$, $K = 6$), ($m = 100$, $K = 30$) and ($m = 100$, $K = 60$). We
measured the statistical correlation between the number of voters and
the quality of the algorithms $C/C_{\ideal}$. The ANOVA test in most
cases showed that there is no such correlation. The only exception was
\textsc{S2} data set, for which we obtained an almost negligible
correlation. For example, for ($m = 10, K = 3$) Algorithm $C$ under
data set \textsc{S2} for Monroe's system for $n = 5000$ gave
$C/C_{\ideal} = 0.88$, while for $n = 100$ (in the previous section)
we got $C/C_{\ideal} = 0.89$. Thus we conclude that in practice the
number of agents has almost no influence on the quality of the results
provided by our algorithms.

Next, we fixed the number of voters $n = 1000$ and the ratio $K/m =
0.3$, and for each $m$ ranging from $30$ to $300$ with the step of
$30$ (naturally, as $m$ changed, so did $K$ to maintain the ratio
$K/m$), we run 50 experiments. We repeated this procedure for $K/m =
0.6$. The relation between $m$ and $C/C_{\ideal}$ for \textsc{Mv} and
\textsc{Ur}, under both Monroe's rule and Chamberlin-Courant's rule,
is given in Figures~\ref{fig:changing_m_monroe}
and~\ref{fig:changing_m_cc} (the results for $K/m = 0.6$ look similar).

Finally, we fixed $n = 1000$ and $m = 100$, and for each $K/m$ ranging
from $0.1$ and $0.5$ with the step of $0.1$ we run $50$ experiments. The
relation between the ratio $K/m$ and the quality $C/C_{\ideal}$ is
presented in Figures~\ref{fig:changing_km_monroe} and~~\ref{fig:changing_km_cc}. 

For the case of Chamberlin-Courant's system increasing the size of the
committee we elect improves agent satisfaction: Since there are no
constraints on the number of agents matched to a given alternative,
larger committees mean more opportunities to satisfy the agents.  For
Monroe, larger committees may lead to lower total satisfaction.  This
happens if many agents like a particular alternative a lot, but only
some of them can be matched to this alternative and others have to be
matched to their less-preferred ones.  Nonetheless, we see that
Algorithm~C achieves $C/C_{\ideal} = 0.925$ even for $K/m = 0.5$ for
the NetFlix data set.

Our conclusions from these experiments are the following.  For
Monroe's rule, even Algorithm~A achieves very good results. However,
Algorithm~C consistently achieves better ones (indeed, almost perfect
ones).  Randomized algorithms consistently do worse than our
deterministic ones.

\begin{figure}[t]
  \centering
  \includegraphics[width=0.75\textwidth]{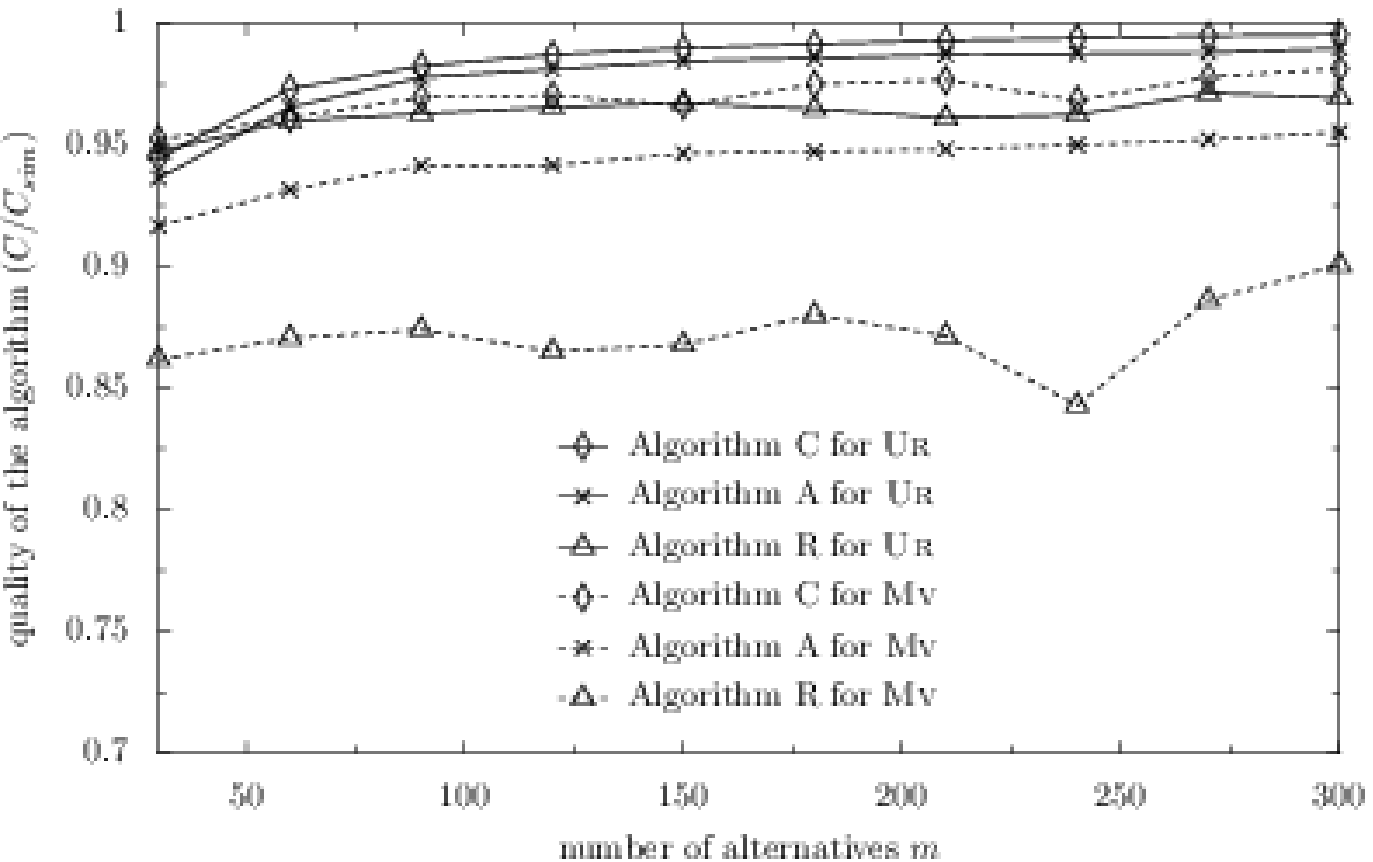}
  \caption{The relation between the number of alternatives
    {$m$} and the quality of the algorithms
    {$C/C_{\ideal}$} for the Monroe's system.}
  \label{fig:changing_m_monroe}
\end{figure}
\begin{figure}[t]
  \centering
  \includegraphics[width=0.75\textwidth]{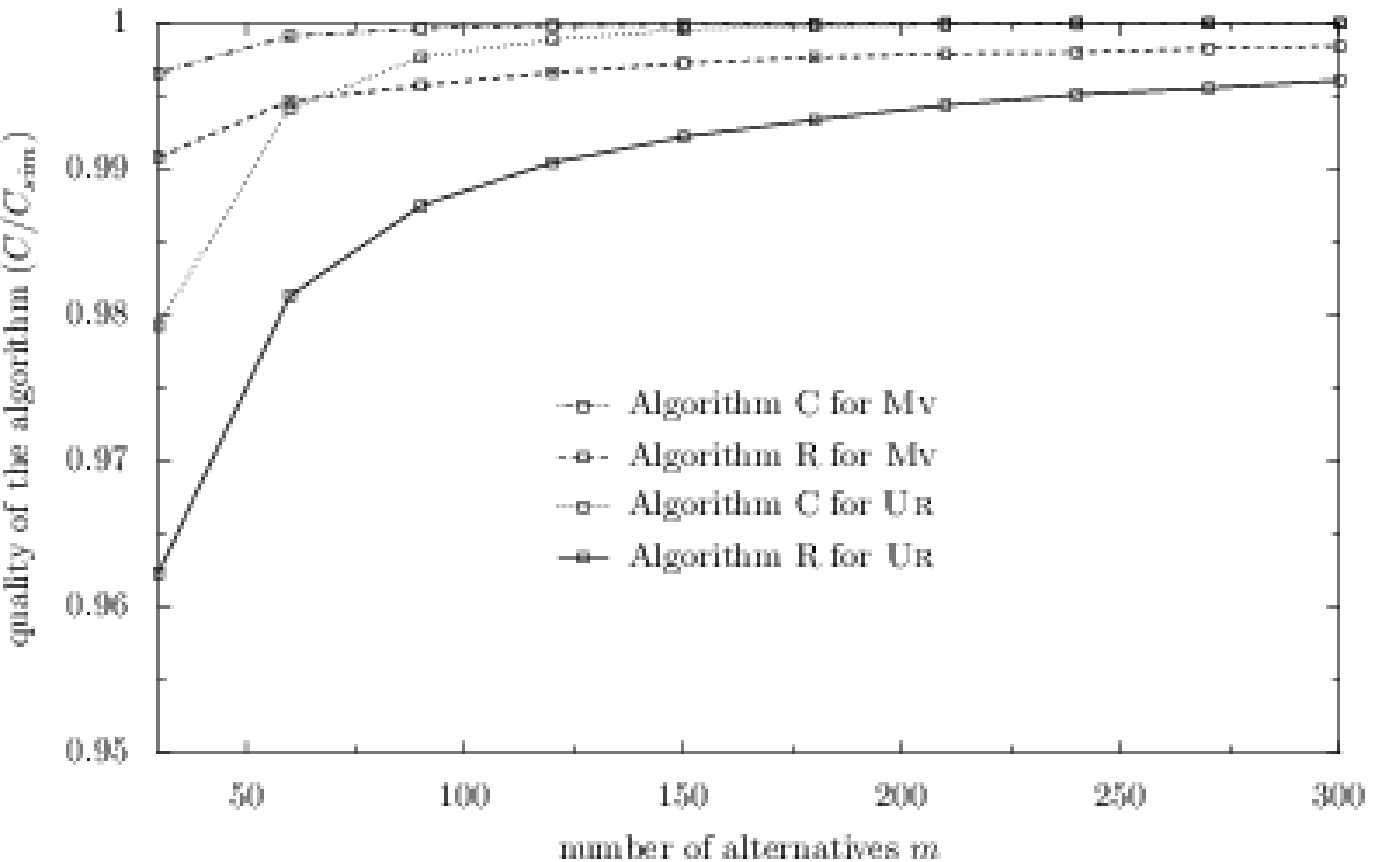}
  \caption{The relation between the number of alternatives
    {$m$} and the quality of the algorithms
    {$C/C_{\ideal}$} for the Chamberlin-Courant's system.}
  \label{fig:changing_m_cc}
\end{figure}

\begin{figure}[t]
  \centering
  \includegraphics[width=0.75\textwidth]{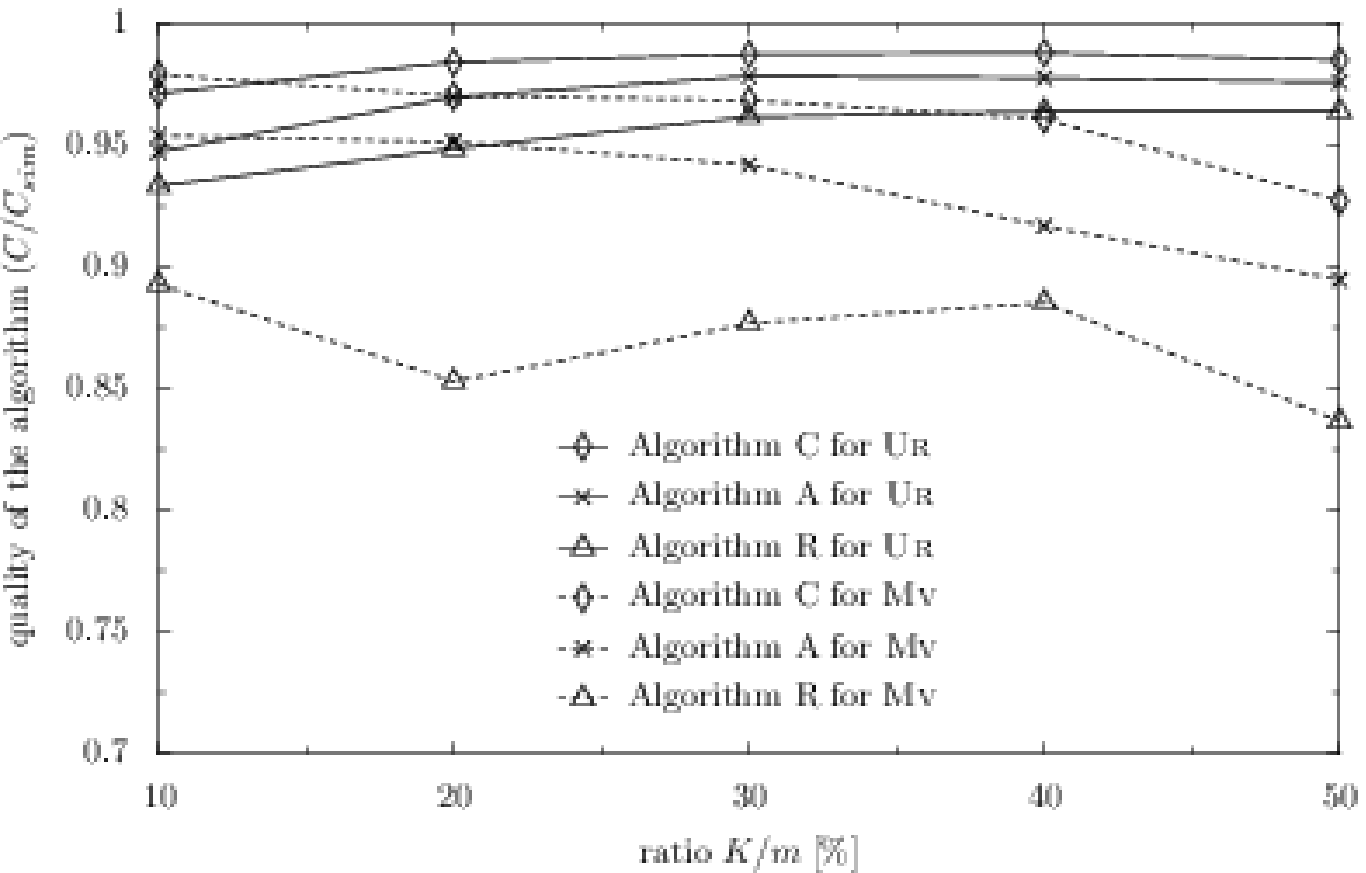}
  \caption{The relation between the ratio {$K/m$} and the
    quality of the algorithms {$C/C_{\ideal}$} for the Monroe's
    system.}
  \label{fig:changing_km_monroe}
\end{figure}
\begin{figure}[t]
  \centering
  \includegraphics[width=0.75\textwidth]{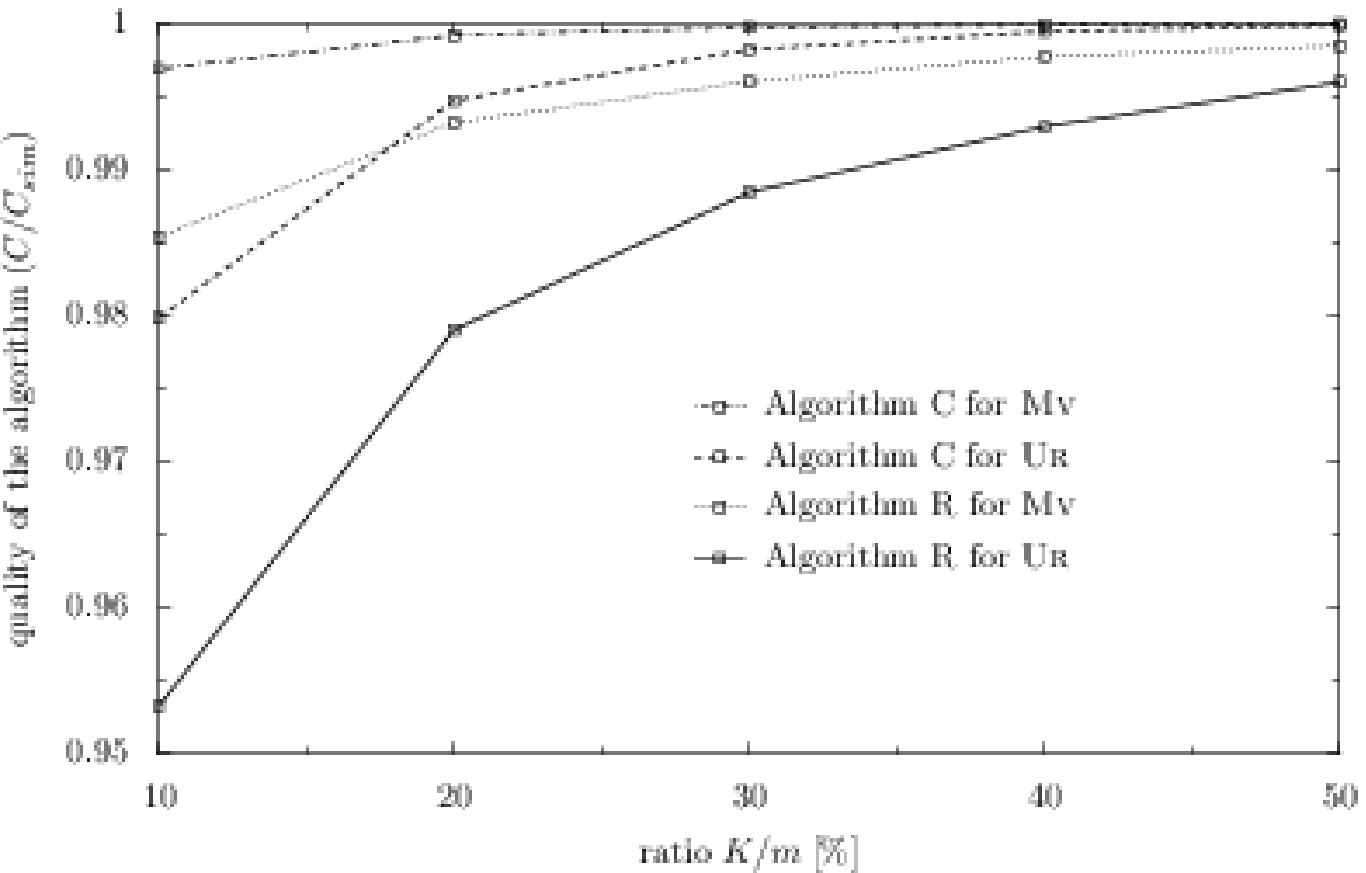}
  \caption{The relation between the ratio {$K/m$} and the
    quality of the algorithms {$C/C_{\ideal}$} for the
    Chamberlin-Courant's system.}
  \label{fig:changing_km_cc}
\end{figure}

\subsection{Running time}

In our final set of experiments we have measured running times of our
algorithms on the data set \textsc{Mv}. We have used a machine with
Intel Pentium Dual T2310 1.46GHz processor and 1.5GB of RAM.  In
Figure~\ref{fig:ilp_runtime} we show the running time of GLPK ILP
solver for Monroe's and for Chamberlin-Courant's rules. These running
times are already large for small instances and they are increasing
exponentially with the number of voters. For Monroe's rule, even for
$K=9, m = 30, n=100$ some of the experiments timed out after 1 hour,
and for $K=9, m = 30, n=200$ none of the experiments finished within
one day.  Thus we conclude that the real application of the ILP
algorithm is very limited. Example running times of the other
algorithms for some combinations of $n$, $m$, and $K$ are presented in
Table~\ref{table:runningTimes}.

\begin{figure}[t]
  \centering
  \includegraphics[width=0.77\textwidth]{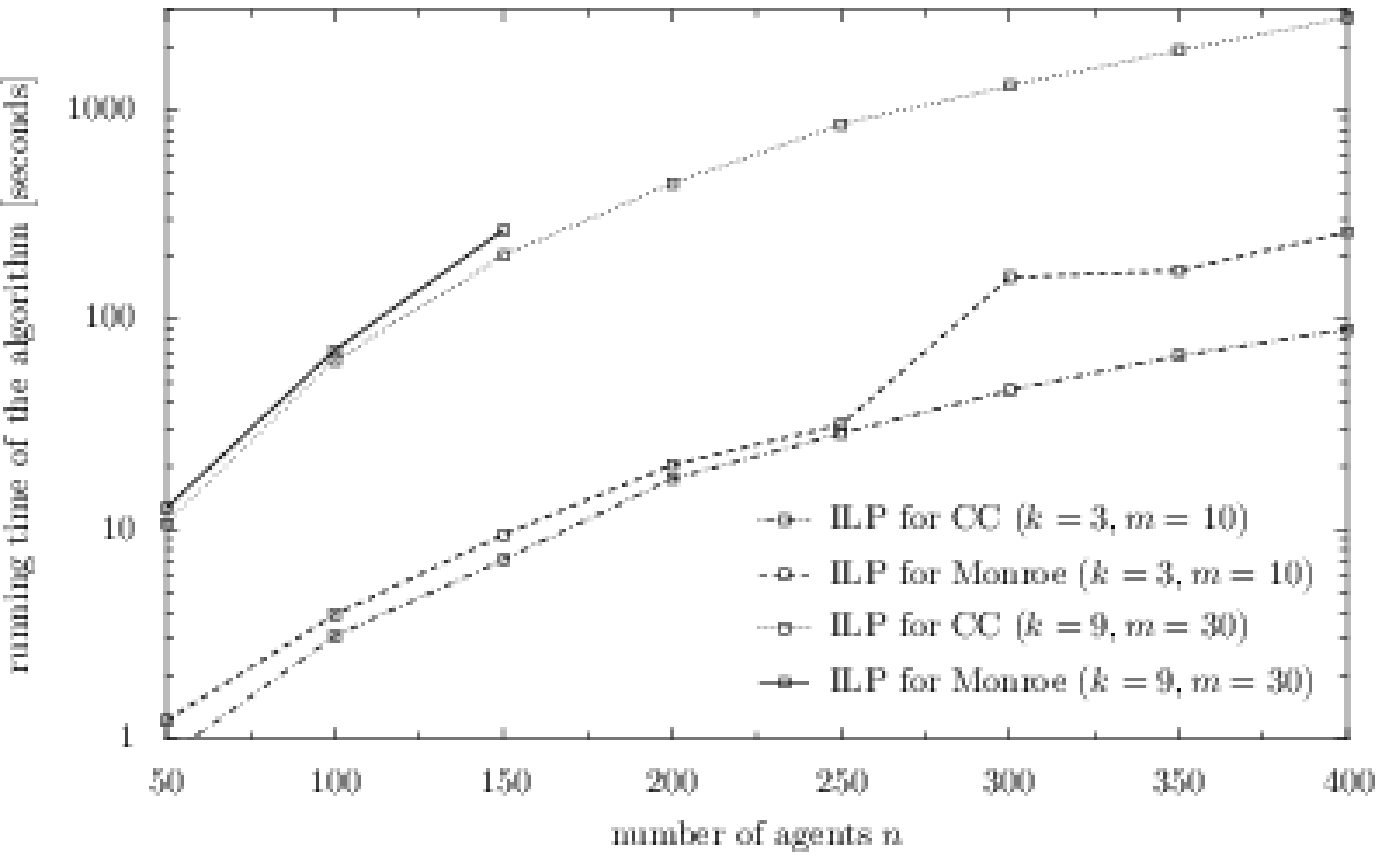}
  \caption{The running time of the standard ILP solver for the
    Monroe's and for the Chamberlin-Courant's systems.  For the
    Monroe's system, for {$K=9, m = 30$}, and for
    {$n \geq 200$} none of the single algorithm execution
    finished within 1 day.}
  \label{fig:ilp_runtime}
\end{figure}

\begin{table}[h!]
\begin{center}
\begin{tabular}{|c|c|c|c|c||c|c|c|}
\cline{3-8}
\multicolumn{2}{c|}{} & \multicolumn{3}{|c||}{$m=10$, $K=3$} & \multicolumn{3}{|c|}{$m=10$, $K=6$} \\
\cline{2-8}
\multicolumn{1}{c|}{} & $n=$ & 2000 & 6000 & 10000 & 2000 & 6000 & 10000\\
\cline{1-8}
\multirow{5}{*}{\begin{sideways} Monroe \end{sideways}}
& A & $0.01$ & $0.03$ & $0.05$  & $0.01$ & $0.04$ & $0.07$ \\
& B & $0.08$ & $0.9$ & $2.3$  & $0.2$ & $1.4$ & $3.6$ \\
& C & $1.1$ & $8$ & $22$  & $2.1$ & $16$ & $37$ \\
& GM & $0.8$ & $7.3$ & $20$ & $1.9$ & $13$ & $52$ \\
& R & $7.6$ & $50$ & $180$  & $6.5$ & $52$ & $140$ \\
\cline{1-8}
\multirow{4}{*}{\begin{sideways} CC \end{sideways}}
& C & $0.02$ & $0.07$ & $0.12$ & $0.05$ & $0.14$ & $0.26$ \\
& GM & $0.003$ & $0.009$ & $0.015$ & $0.003$ & $0.01$ & $0.018$ \\
& P & $0.009$ & $0.032$ & $0.05$ & $0.008$ & $0.02$ & $0.05$ \\
& R & $0.014$ & $0.04$ & $0.065$ & $0.02$ & $0.06$ & $0.11$ \\
\cline{1-8}
\multicolumn{2}{c}{} \\

\cline{3-8}
\multicolumn{2}{c|}{} & \multicolumn{3}{|c||}{$m=100$, $K=30$} & \multicolumn{3}{|c|}{$m=100$, $K=60$} \\
\cline{2-8}
\multicolumn{1}{c|}{} & $n=$ & 2000 & 6000 & 10000 & 2000 & 6000 & 10000\\
\cline{1-8}
\multirow{5}{*}{\begin{sideways} Monroe \end{sideways}}
& A & $0.5$ & $1.6$ & $2.8$  & $0.9$ & $2.8$ & $4.9$ \\
& B & $0.8$ & $4$ & $9.5$  & $1.7$ & $8$ & $18$ \\
& C & $38$ & $140$ & $299$  & $64$ & $221$ & $419$ \\
& GM & $343$ & $2172$ & $5313$ & $929$ & $5107$ & $13420$ \\
& R & $41$ & $329$ & $830$  & $88$ & $608$ & $1661$ \\
\cline{1-8}
\multirow{4}{*}{\begin{sideways} CC \end{sideways}}
& C & $4.3$ & $11$ & $19$ & $7.5$ & $19$ & $31$ \\
& GM & $0.06$ & $0.2$ & $0.4$ & $0.09$ & $0.3$ & $0.7$ \\
& P & $0.03$ & $0.1$ & $0.26$ & $0.03$ & $0.1$ & $0.2$ \\
& R & $0.06$ & $0.24$ & $0.45$ & $0.1$ & $0.4$ & $0.8$ \\
\cline{1-8}
\end{tabular}
\end{center}

\caption{Example running times of the algorithms [in seconds].}
\label{table:runningTimes}
\end{table}

\section{Conclusions}\label{sec:conclusions}
We have provided experimental evaluation of a number of algorithms
(both known ones and their extensions) for computing the winners under
Monroe's rule and under Chamberlin-Courant's rule. While finding
winners under these rules is
$\np$-hard~\cite{pro-ros-zoh:j:proportional-representation,budgetSocialChoice,fullyProportionalRepr},
it turned out that in practice we can obtain very high quality
solutions using simple algorithms. Indeed, both for Monroe's rule and
for Chamberlin-Courant's rule we recommend using Algorithm~C (or
Algorithm~A on very large Monroe elections).
We believe that our results mean that (approximations of) Monroe's and
Chamberlin-Courant's rules can be used in practice.

\smallskip
\noindent\textbf{Acknowledgements} The
authors were supported in part by AGH Univ. grant 11.11.120.865, by
the Foundation for Polish Science's Homing/Powroty program, by
Poland's National Science Center grant DEC-2011/03/B/ST6/01393, and by
EU's Human Capital Program "National PhD Programme in Mathematical
Sciences" carried out at the University of Warsaw.

\bibliographystyle{abbrv}
\bibliography{main,grypiotr2006}

\begin{thebibliography}{10}

\bibitem{bpublicchoice85}
S.~Berg.
\newblock {Paradox of voting under an urn model: The effect of homogeneity}.
\newblock {\em Public Choice}, 47:377--387, 1985.

\bibitem{fullyProportionalRepr}
N.~Betzler, A.~Slinko, and J.~Uhlmann.
\newblock On the computation of fully proportional representation.
\newblock Technical report, U. of Auckland, November 2011.

\bibitem{cha-cou:j:cc}
B.~Chamberlin and P.~Courant.
\newblock Representative deliberations and representative decisions:
  Proportional representation and the {B}orda rule.
\newblock {\em American Political Science Review}, 77(3):718--733, 1983.

\bibitem{Kamishima:Nantonac}
T.~Kamishima.
\newblock Nantonac collaborative filtering: recommendation based on order
  responses.
\newblock In {\em Proceedings of KDD-03}, pages 583--588, 2003.

\bibitem{budgetSocialChoice}
T.~Lu and C.~Boutilier.
\newblock Budgeted social choice: {F}rom consensus to personalized decision
  making.
\newblock In {\em Proceedings of IJCAI-2011}, pages 280--286, 2011.

\bibitem{mallowImplementation2011}
T.~Lu and C.~Boutilier.
\newblock Learning {M}allows models with pairwise preferences.
\newblock In {\em Proceedings of ICML-11}, pages 145--152, June 2011.

\bibitem{mallowModel}
C.~L. Mallows.
\newblock Non-null ranking models. i.
\newblock {\em Biometrika}, 44(1-2):114--130, June 1957.

\bibitem{Mattei:Netflix}
N.~Mattei, J.~Forshee, and J.~Goldsmith.
\newblock An empirical study of voting rules and manipulation with large
  datasets.
\newblock In {\em COMSOC}, 2012.

\bibitem{mei-pro-ros-zoh:j:multiwinner}
R.~Meir, A.~Procaccia, J.~Rosenschein, and A.~Zohar.
\newblock The complexity of strategic behavior in multi-winner elections.
\newblock {\em JAIR}, 33:149--178, 2008.

\bibitem{mon:j:monroe}
B.~Monroe.
\newblock Fully proportional representation.
\newblock {\em American Political Science Review}, 89(4):925--940, 1995.

\bibitem{submodular}
G.~Nemhauser, L.~Wolsey, and M.~Fisher.
\newblock An analysis of approximations for maximizing submodular set
  functions.
\newblock {\em Mathematical Programming}, 14(1):265--294, 1978.

\bibitem{ore:p:cc}
J.~Oren.
\newblock Personal communication, 2012.

\bibitem{potthoff-brams}
R.~Potthoff and S.~Brams.
\newblock Proportional representation: {B}roadening the options.
\newblock {\em Journal of Theoretical Politics}, 10(2):147--178, 1998.

\bibitem{pro-ros-zoh:j:proportional-representation}
A.~Procaccia, J.~Rosenschein, and A.~Zohar.
\newblock On the complexity of achieving proportional representation.
\newblock {\em Social Choice and Welfare}, 30(3):353--362, 2008.

\bibitem{sko-fal-sli:w:multiwinner}
P.~Skowron, P.~Faliszewski, and A.~Slinko.
\newblock Fully proportional representation as resource allocation:
  Approximability results.
\newblock Technical Report arXiv:0809.4484~[cs.GT], arXiv.org, Aug. 2012.

\bibitem{Walsh11}
T.~Walsh.
\newblock Where are the hard manipulation problems?
\newblock {\em JAIR}, 42:1--29, 2011.

\end{thebibliography}

\end{document}